\documentclass[a4paper]{article}

\usepackage[english]{babel}
\usepackage[utf8]{inputenc}
\usepackage[T1]{fontenc}
\usepackage{amsmath}
\usepackage{graphicx}
\usepackage[colorinlistoftodos]{todonotes}
\usepackage[a4paper,left=2.5cm,right=2.5cm,top=2.cm,bottom=2cm]{geometry}
\usepackage{authblk}
\usepackage{txfonts}
\usepackage{bbold} 
\usepackage{dsfont}
\usepackage{amsopn}  
\usepackage{amssymb}
\usepackage{graphicx}
\newtheorem{example}{Example}
\newtheorem{proposition}{Proposition}
\newtheorem{definition}{Definition}
\newtheorem{lemma}{Lemma}
\newtheorem{corollary}{Corollary}
\newenvironment{proof}{\noindent\textit{Proof~~}}
{\nolinebreak[4]\hfill$\blacksquare$\\\par}
\title{Strong isomorphism in Marinatto-Weber type quantum games}

\author{Piotr Fr\k{a}ckiewicz}
\affil{Institute of Mathematics\\ Pomeranian University, Poland}
\date{\today}

\begin{document}
\maketitle

\begin{abstract}
Our purpose is to focus attention on a new criterion for quantum schemes by bringing together the notions of quantum game and game isomorphism. A quantum game scheme is required to generate the classical game as a special case. Now, given a quantum game scheme and two isomorphic classical games, we additionally require the resulting quantum games to be isomorphic as well. We show how this isomorphism condition influences the players' strategy sets. We are concerned with the Marinatto-Weber type quantum game scheme and the strong isomorphism between games in strategic form. 
\end{abstract}
\section{Introduction}

The Marinatto-Weber  (MW) scheme introduced in \cite{marinatto} is a straightforward way to apply the power of quantum mechanics to classical game theory. In the simplest case of $2\times 2$ games, the players manipulate their own qubits of a two-qubit state either with the identity $\mathds{1}$ or the Pauli operator $\sigma_{x}$. Therefore, it has found application in many other branches of game theory: from evolutionary game theory \cite{iqbal2}, \cite{nawaz2} to extensive-form games \cite{fracorultimatumgame} and duopoly examples \cite{iqbalduopoly}, \cite{khanbertrand}. In paper~\cite{fracormw} we pointed out a few undesirable properties of the MW scheme and introduced a refined quantum game model. 

Though it is possible to extend both the MW scheme and our refinement to consider more complex games than $2 \times 2$, possible generalizations can be defined in many different ways. A result concerning $3\times 3$ games can be found in \cite{iqbal2} and \cite{iqbal3x3}. The authors proposed suitable three-element sets of players' strategies to obtain a generalized $3\times 3$ game. On the other hand, our work \cite{fracorpolonica} provides another way to define players' strategy sets that remains valid for any finite $n\times m$ games.

Certainly, one can find yet other ways to generalize the MW scheme. Hence it would be interesting to place additional restrictions on a quantum game scheme and examine how they refine the quantum model. In this paper we formulate a criterion in terms of isomorphic games. Given two isomorphic games we require the corresponding quantum games to be isomorphic as well. If, for example, two bimatrix games differ only in the order of players' strategies, they describe the same problem from a game-theoretical point of view. Given a quantum scheme, it appears reasonable to assume that the resulting quantum game will not depend on the numbering of players’ strategies in the classical game.
\section{Preliminaries}
\subsection{Marinatto-weber type quantum game scheme}
In paper \cite{fracormw} and \cite{fracornewscheme} we presented a refinement of the Marinatto-Weber scheme \cite{marinatto}. The motivation of constructing our scheme was twofold. Our model enables the players to choose between playing a fixed quantum strategy and  classical strategies. The second aim was to construct the scheme that generates the classical game by manipulating the players' strategies rather than the initial quantum state. In what follows, we recall the scheme for the case of $2\times 2$ bimatrix game,
\begin{equation}\label{2x2}
\bordermatrix{& l & r \cr t & (a_{00}, b_{00}) & (a_{01}, b_{01}) \cr b & (a_{10}, b_{10}) & (a_{11}, b_{11})}, ~\mbox{where}~(a_{ij}, b_{ij})\in \mathbb{R}.  
\end{equation}
\begin{definition}
The quantum scheme for game~(\ref{2x2}) is defined on an inner product space $(\mathbb{C}^2)^{\otimes 4}$ by the triple
\begin{equation}\label{schememoj}
\Gamma_{Q} = (H, (S_{1}, S_{2}), (M_{1}, M_{2})),
\end{equation}
where
\begin{itemize}
\item $H$ is a positive operator,
\begin{equation}\label{positiveoperator1}
H = (\mathds{1}\otimes \mathds{1} - |11\rangle \langle 11|)\otimes |00\rangle \langle 00| + |11\rangle \langle 11|\otimes |\Psi\rangle \langle \Psi|,
\end{equation}
and
\begin{equation} 
|\Psi\rangle = \alpha|00\rangle + \beta|01\rangle + \gamma|10\rangle + \delta|11\rangle \in \mathbb{C}^2\otimes \mathbb{C}^2
\end{equation}
such that $\| |\Psi \rangle \| = 1$, \item $S_{1} = \left\{P^{(1)}_{i}\otimes U^{(3)}_{j}, i,j=0,1\right\}$, $S_{2} = \left\{P^{(2)}_{k}\otimes U^{(4)}_{l}, k,l=0,1\right\}$ are the players' strategy sets, and the upper indices identify the subspace $\mathbb{C}^2$ of $(\mathbb{C}^2)^{\otimes 4}$ on which the operators
\begin{equation}\label{operatory}
P_{0} = |0\rangle \langle 0|,~ P_{1} = |1\rangle \langle 1|, \quad U_{0} = \mathds{1},~ U_{1} = \sigma_{x},
\end{equation}
are defined,
\item $M_{1}$ and $M_{2}$ are the measurement operators 
\begin{equation}\label{payoffmeasurement}
M_{1(2)} = \mathds{1}\otimes \mathds{1}\otimes \left(\sum_{x,y = 0,1}a_{xy}(b_{xy})|xy\rangle \langle xy|\right)
\end{equation}
that depend on the payoffs~$a_{xy}$ and $b_{xy}$ from (\ref{2x2}).
\end{itemize}
\end{definition}
The scheme proceeds in the similar way as the MW scheme--the players determine the final state by choosing their strategies and acting on operator $H$. As a result, they determine the following density operator:
\begin{align}\label{pierwszydensity}
\rho_{f} &= \left(P^{(1)}_{i}\otimes P^{(2)}_{k}\otimes U^{(3)}_{j}\otimes U^{(4)}_{l}\right)H \left(P^{(1)}_{i}\otimes P^{(2)}_{k}\otimes U^{(3)}_{j}\otimes U^{(4)}_{l}\right)\nonumber \\ &=\begin{cases}|11\rangle \langle 11|\otimes \left(U^{(3)}_{j}\otimes U^{(4)}_{l}|\Psi\rangle \langle \Psi| U^{(3)}_{j}\otimes U^{(4)}_{l}\right) &\mbox{if}~i=j=1, \\ |ij\rangle \langle ij|\otimes \left(U^{(3)}_{j}\otimes U^{(4)}_{l}|00\rangle \langle 00|U^{(3)}_{j}\otimes U^{(4)}_{l}\right) &\mbox{if otherwise}. \end{cases}
\end{align}
Next, the payoffs for player 1 and 2 are
\begin{equation}
\mathrm{tr}(\rho_{f}M_{1})~~\mbox{and}~~\mathrm{tr}(\rho_{f}M_{2}).
\end{equation}
As it was shown in \cite{fracornewscheme}, scheme (\ref{schememoj}) can be summarized by the following matrix game
\begin{equation}\label{matrixmoj}
\bordermatrix{& P^{(2)}_{0}\otimes \mathds{1}^{(4)} & P^{(2)}_{0}\otimes \sigma_{x}^{(4)} & P^{(2)}_{1}\otimes \mathds{1}^{(4)}  & P^{(2)}_{1}\otimes \sigma_{x}^{(4)} \cr P^{(1)}_{0}\otimes \mathds{1}^{(3)} & X_{00} & X_{01} & X_{00} & X_{01} \cr P^{(1)}_{0}\otimes \sigma_{x}^{(3)} & X_{10} & X_{11} & X_{10} & X_{11} \cr P^{(1)}_{1}\otimes \mathds{1}^{(4)} & X_{00} & X_{01} & \Delta_{00} & \Delta_{01} \cr P^{(1)}_{1}\otimes \sigma_{x}^{(4)} & X_{10} & X_{11} & \Delta_{10} & \Delta_{11}},
\end{equation}
where
\begin{align}\label{legenda0}\begin{split}
&X_{ij} = (a_{ij}, b_{ij}), ~~\mbox{for}~~ i,j = 0,1 \cr &\Delta_{00} = |\alpha|^2X_{00} + |\beta|^2X_{01} + |\gamma|^2X_{10} + |\delta|^2X_{11}, \cr &\Delta_{01} = |\alpha|^2X_{01} + |\beta|^2X_{00} + |\gamma|^2X_{11} + |\delta|^2X_{10}, \cr &\Delta_{10} = |\alpha|^2X_{10} + |\beta|^2X_{11} + |\gamma|^2X_{00} + |\delta|^2X_{01}, \cr &\Delta_{11} = |\alpha|^2X_{11} + |\beta|^2X_{10} + |\gamma|^2X_{01} + |\delta|^2X_{00}.  \end{split}
\end{align}
\subsection{Strong isomorphism}
The notion of strong isomorphism defines classes of games that are the same up to the numbering of the players and the order of players' strategies.
The following definitions are taken from \cite{garcia} (see also \cite{nash}, \cite{peleg1} and \cite{peleg2}). The first one defines a mapping that associates players and their actions in one game with players and their actions in the other game.
\begin{definition}
Let $\Gamma = (N, (S_{i})_{i\in N}, (u_{i})_{i\in N})$ and $\Gamma' = (N, (S'_{i})_{i\in N}, (u'_{i})_{i\in N})$ be games in strategic form. A game mapping $f$ from $\Gamma$ to $\Gamma'$ is a tuple $f = (\eta, (\varphi_{i})_{i\in N})$, where $\eta$ is a bijection from $N$ to $N$ and for any $i\in N$, $\varphi_{i}$~is a bijection from $S_{i}$ to $S'_{\eta(i)}$.
\end{definition} 
In general case, the mapping $f$ from $(N, (S_{i})_{i\in N}, (u_{i})_{i\in N})$ to $(N, (S'_{i})_{i\in N}, (u'_{i})_{i\in N})$ identifies player $i\in N$ with player $\eta(i)$ and maps $S_{i}$ to $S_{\eta(i)}$. This means that a strategy profile $(s_{1},\dots, s_{n}) \in S_{1} \times \dots \times S_{n}$ is mapped into profile $(s'_{1}, \dots, s'_{n})$ that satisfies equation $s'_{\eta(i)} = \varphi_{i}(s_{i})$ for $i\in N$.

The notion of game mapping is a basis for the definition of game isomorphism. Depending on how rich structure of the game is to be preserved we can distinguish various types of game isomorphism. One that preserves the players' payoff functions is called a strong isomorphism. The formal definition is as follows:
\begin{definition}\label{izo}
Given two strategic games $\Gamma = (N,(S_{i})_{i\in N}, (u_{i})_{i\in N})$ and $\Gamma' = (N, (S'_{i})_{i\in N}, (u'_{i})_{i\in N})$, a game mapping $f = (\eta, (\varphi_{i})_{i\in N})$ is called a strong isomorphism if relation $u_{i}(s) = u'_{\eta(i)}(f(s))$ holds for each $i\in N$ and each strategy profile $s\in S_{1}\times \dots \times S_{n}$. 
\end{definition} 
From the above definition it may be concluded that if there is a strong isomorphism between games $\Gamma$ and $\Gamma'$, they may differ merely by the numbering of players and the order of their strategies.

The following lemma shows that relabeling players and their strategies do not affect the game with regard to Nash equilibria. If $f$ is a strong isomorphism between games $\Gamma$ and $\Gamma'$, one may expect that the Nash equilibria in $\Gamma$ map to ones in $\Gamma'$ under $f$. 
\begin{lemma}\label{lemma1}
Let $f$ be a strong isomorphism between games $\Gamma$ and $\Gamma'$. Strategy profile $s^{*} = (s^*_{1}, \dots, s^*_{n}) \in S_{1} \times \dots \times S_{n}$ is a Nash equilibrium in game $\Gamma$ if and only if $f(s^*) \in S'_{1}\times \dots \times S'_{n}$ is a Nash equilibrium in $\Gamma'$. 
\end{lemma}
\section{Application of game isomorphism to
Marinatto-Weber type quantum game schemes}
It is not hard to see that we can define a wide variety of schemes based on the MW approach. We can modify operator (\ref{positiveoperator1}) and the players' strategies to construct another scheme still satisfying the requirement about generalization of the input game. The following example of such a scheme is particularly interesting. 

Let us consider a triple
\begin{equation}\label{schemecorrelated}
\Gamma'_{Q} = (H',(S'_{1}, S'_{2}), (M_{1}, M_{2}))
\end{equation}
with the components defined as follows:
\begin{itemize}
\item $H'$ is a positive operator,
\begin{equation}
H' = |00\rangle \langle 00| \otimes |00\rangle \langle 00| + |01\rangle \langle 01| \otimes |0\rangle \langle 0| \otimes \rho_{2} + |10\rangle \langle 10| \otimes \rho_{1} \otimes |0\rangle \langle 0| + |11\rangle \langle 11|\otimes |\Psi\rangle \langle \Psi|,
\end{equation}
where $|\Psi\rangle\in \mathbb{C}^2\otimes \mathbb{C}^2$ such that $\| |\Psi \rangle \| = 1$, $\rho_{1}$ and $\rho_{2}$ are the reduced density operators of $|\Psi \rangle \langle \Psi|$, i.e., $\rho_{1} = \mathrm{tr}_{2}(|\Psi \rangle \langle \Psi|)$ and $\rho_{2} = \mathrm{tr}_{1}(|\Psi \rangle \langle \Psi|)$,
\item $S'_{1} = \left\{P^{(1)}_{0}\otimes \mathds{1}^{(3)}, P^{(1)}_{0}\otimes \sigma^{(3)}_{x}, P^{(1)}_{1}\otimes \mathds{1}^{(3)}\right\}$ and $S'_{2} = \left\{P^{(2)}_{0}\otimes \mathds{1}^{(4)}, P^{(2)}_{0}\otimes \sigma^{(4)}_{x}, P^{(2)}_{1}\otimes \mathds{1}^{(4)}\right\}$ are the players' strategy sets, \item $M_{1}$ and $M_{2}$ are the measurement operators defined by equation~(\ref{payoffmeasurement}).
\end{itemize}
It is immediate that the resulting final state $\rho'_{f}$ is a density operator for each (pure or mixed) strategy profile. For example, player 1's strategy $P^{(1)}_{0}\otimes \sigma^{(3)}_{x}$ and player 2's strategy $P^{(2)}_{1}\otimes \mathds{1}^{(4)}$ imply
\begin{align}\label{rhoprzyklad}
\rho'_{f} = \left(P^{(1)}_{0}\otimes  P^{(2)}_{1}\otimes \sigma^{(3)}_{x}\otimes \mathds{1}^{(4)}\right)H'\left(P^{(1)}_{0}\otimes  P^{(2)}_{1}\otimes \sigma^{(3)}_{x}\otimes \mathds{1}^{(4)}\right) = |01\rangle \langle 01|\otimes |1\rangle \langle 1| \otimes \rho_{2}.
\end{align}
As a result, the players' payoff functifons $u'_{1}$ and $u'_{2}$ given by $\mathrm{tr}(\rho'_{f}M_{1})$ and $\mathrm{tr}(\rho'_{f}M_{2})$, respectively, are well-defined. It is also clear that scheme~(\ref{schemecorrelated}) produces the classical game in a similar way to scheme~(\ref{schememoj}). The players play the classical game as long as they choose the strategies $P_{0}\otimes \mathds{1}$ and $P_{0}\otimes \sigma_{x}$. This can be seen by determining $\mathrm{tr}(\rho'_{f}M_{1(2)})$ for each strategy profile and arranging the obtained values into a matrix. As an example, let us determine $\mathrm{tr}(\rho'_{f}M_{1(2)})$ for the final state $\rho'_{f}$ given by~(\ref{rhoprzyklad}). Let $|\Psi\rangle$ represent a general two qubit state, 
\begin{equation}
|\Psi\rangle = \alpha|00\rangle + \beta|01\rangle + \gamma|10\rangle + \delta|11\rangle. 
\end{equation}
Since
\begin{align}\begin{split}
\rho_{1} = (|\alpha|^2 + |\beta|^2)|0\rangle \langle 0| + (\alpha\gamma^{*} + \beta\delta^*)|0\rangle \langle 1| + (\gamma\alpha^{*} + \delta\beta^*)|1\rangle \langle 0| + (|\gamma|^2 + |\delta|^2)|1\rangle \langle 1|, \\  \rho_{2} = (|\alpha|^2 + |\gamma|^2)|0\rangle \langle 0| + (\alpha\beta^{*} + \gamma\delta^*)|0\rangle \langle 1| + (\beta\alpha^{*} + \delta\gamma^*)|1\rangle \langle 0| + (|\beta|^2 + |\delta|^2)|1\rangle \langle 1|, \end{split}
\end{align}
the players' strategies $P^{(1)}_{0}\otimes \sigma^{(3)}_{x}$ and $P^{(2)}_{1}\otimes \mathds{1}^{(4)}$ generate the following form of the final state:
\begin{equation}
\rho'_{f} = |01\rangle \langle 01| \otimes \Bigl((|\alpha|^2 + |\gamma|^2)|10\rangle \langle 10| + (\alpha\beta^* + \gamma\delta^*)|10\rangle \langle 11| + (\beta\alpha^* + \delta\gamma^*)|11\rangle \langle 10| + (|\beta|^2 + |\delta|^2)|11\rangle \langle 11|\Bigr).
\end{equation}
Hence
\begin{equation}
(\mathrm{tr}(\rho'_{f}M_{1}), \mathrm{tr}(\rho'_{f}M_{2}) )= (|\alpha|^2 + |\gamma|^2)(a_{10}, b_{10}) + (|\beta|^2 + |\delta|^2)(a_{11}, b_{11}).
\end{equation}
The values  $(\mathrm{tr}(\rho'_{f}M_{1}), \mathrm{tr}(\rho'_{f}M_{2}) )$ for all strategy combinations are given by the following matrix:
\begin{equation}\label{macierzfake}
\bordermatrix{& P^{(2)}_{0}\otimes \mathds{1}^{(4)} & P^{(2)}_{0}\otimes \sigma_{x}^{(4)} & P^{(2)}_{1}\otimes \mathds{1}^{(4)}\cr P^{(1)}_{0}\otimes \mathds{1}^{(3)} & X_{00} & X_{01} & \Delta_{02} \cr P^{(1)}_{0}\otimes \sigma_{x}^{(3)} & X_{10} & X_{11} & \Delta_{12} \cr P^{(1)}_{1}\otimes \mathds{1}^{(3)} & \Delta_{20} & \Delta_{21} & \Delta_{22}}
\end{equation}
where
\begin{align}\label{legenda1}
&X_{ij} = (a_{ij}, b_{ij}), ~~\mbox{for}~~ i,j = 0,1 \cr &\Delta_{02} = (|\alpha|^2 + |\gamma|^2)X_{00} + (|\beta|^2 + |\delta|^2)X_{01}; \cr &\Delta_{12} = (|\alpha|^2 + |\gamma|^2)X_{10} + (|\beta|^2 + |\delta|^2)X_{11};\cr &\Delta_{20} = (|\alpha|^2 + |\beta|^2)X_{00} + (|\gamma|^2 + |\delta|^2)X_{10}; \cr &\Delta_{21} = (|\alpha|^2 + |\beta|^2)X_{01} + (|\gamma|^2 + |\delta|^2)X_{11}; \cr &\Delta_{22} =  |\alpha|^2X_{00} + |\beta|^2X_{01} + |\gamma|^2X_{10} + |\delta|^2X_{11}.
\end{align}
It follows easily that matrix game~(\ref{macierzfake}) is a genuine extension of (\ref{2x2}). Although payoff profiles $\Delta_{ij} \ne \Delta_{22}$ are also achievable in (\ref{2x2}), the players, in general, are not able to obtain $\Delta_{22}$ when choosing their (mixed) strategies.

To sum up, scheme (\ref{schemecorrelated}) might seem to be acceptable as long as scheme (\ref{schememoj}) is acceptable. Matrix game (\ref{macierzfake}) includes (\ref{2x2}) and depending on the initial state $|\Psi\rangle$ it may give extraordinary Nash equilibria. It is worth pointing out that the Nash equilibria in (\ref{macierzfake}) correspond to correlated equilibria in (\ref{2x2}), (see \cite{fracorgv}). However scheme  (\ref{schemecorrelated}) fails to imply the isomorphic games when the input games are isomorphic. We can make this clear with the following example.
\begin{example}\label{1przyklad}
\textup{Let us consider the game of ``Chicken" $\Gamma_1$ and its (strongly) isomorphic counterpart $\Gamma_{2}$,}
\begin{equation}\label{gameschicken}
\Gamma_{1}\colon \;\bordermatrix{& l & r \cr  t & (6,6) & (2,7) \cr b & (7,2) & (0,0)}, \quad \Gamma_{2}\colon \; \bordermatrix{& l' & r' \cr  t' & (2,7) & (6,6) \cr b' & (0,0) & (7,2)}.
\end{equation}
\textup{The corresponding isomorphism $f = (\pi, \varphi_{1}, \varphi_{2})$ is defined by components}
\begin{equation}
\pi(i) = i~~\mbox{\textup{for}}~~ i=1,2, \quad \varphi_{1} = (t \to t', b \to b'), \varphi_{2} = (l \to r', r \to l'). 
\end{equation}
\textup{Set $|\Psi\rangle = (|00\rangle + |01\rangle + |10\rangle)/\sqrt{3}$. Using (\ref{matrixmoj})  we can write quantum approach (\ref{schememoj}) to games (\ref{gameschicken}) as }
\begin{align}\label{przyklad10}
\Gamma_{Q1}\colon \;\bordermatrix{& P^{(2)}_{0}\otimes \mathds{1}^{(4)} & P^{(2)}_{0}\otimes \sigma_{x}^{(4)} & P^{(2)}_{1}\otimes \mathds{1}^{(4)}  & P^{(2)}_{1}\otimes \sigma_{x}^{(4)} \cr P^{(1)}_{0}\otimes \mathds{1}^{(3)} & (6,6) & (2,7) & (6,6) & (2,7) \cr P^{(1)}_{0}\otimes \sigma_{x}^{(3)} & (7,2) & (0,0) & (7,2) & (0,0) \cr P^{(1)}_{1}\otimes \mathds{1}^{(4)} & (6,6) & (2,7) & (5,5) & (2\frac{2}{3}, 4\frac{1}{3}) \cr P^{(1)}_{1}\otimes \sigma_{x}^{(4)} & (7,2) & (0,0) & (4\frac{1}{3}, 2\frac{2}{3}) & (3,3)} \end{align}
\textup{and}
\begin{align}\label{przyklad101}
\Gamma_{Q2}\colon \;\bordermatrix{& P^{(2)}_{0}\otimes \mathds{1}^{(4)} & P^{(2)}_{0}\otimes \sigma_{x}^{(4)} & P^{(2)}_{1}\otimes \mathds{1}^{(4)}  & P^{(2)}_{1}\otimes \sigma_{x}^{(4)} \cr P^{(1)}_{0}\otimes \mathds{1}^{(3)} & (2,7) & (6,6) & (2,7) & (6,6) \cr P^{(1)}_{0}\otimes \sigma_{x}^{(3)} & (0,0) & (7,2) & (0,0) & (7,2) \cr P^{(1)}_{1}\otimes \mathds{1}^{(4)} & (2,7) & (6,6) & (2\frac{2}{3}, 4\frac{1}{3}) & (5, 5) \cr P^{(1)}_{1}\otimes \sigma_{x}^{(4)} & (0,0) & (7,2) & (3,3) & (4\frac{1}{3}, 2\frac{2}{3})}
\end{align}
\textup{It is fairly easy to see that games (\ref{przyklad10}) and (\ref{przyklad101}) differ in the order of the first two strategies and the second two strategies of player 2. Thus, the games are strongly isomorphic. More formally, one can check that a game mapping $\tilde{f} = (\eta, \tilde{\varphi}_{1}, \tilde{\varphi}_{2})$, where }
\begin{align}\begin{split}
&\tilde{\varphi}_{1} = \left(P^{(1)}_{i}\otimes \mathds{1}^{(3)} \to P^{(1)}_{i}\otimes \mathds{1}^{(3)}, P^{(1)}_{i}\otimes \sigma^{(3)}_{x} \to P^{(1)}_{i}\otimes \sigma^{(3)}_{x}\right),\\
&\tilde{\varphi}_{1} = \left(P^{(2)}_{k}\otimes \mathds{1}^{(4)} \to P^{(2)}_{k}\otimes \sigma^{(4)}_{x}, P^{(2)}_{k}\otimes \sigma^{(4)}_{x} \to P^{(2)}_{k}\otimes \mathds{1}^{(4)}\right), \end{split}
\end{align}
\textup{for $i,k =0,1$ is a strong isomorphism. }

\textup{In the next section we prove a more general result about scheme (\ref{schememoj})). }

\textup{Let us now consider scheme (\ref{schemecorrelated}). Matrix~(\ref{macierzfake}) in terms of input games~(\ref{gameschicken}) implies}
\begin{align}\label{gryschemefake}
\Gamma'_{Q1}\colon \; \bordermatrix{& P^{(2)}_{0}\otimes \mathds{1}^{(4)} & P^{(2)}_{0}\otimes \sigma_{x}^{(4)} & P^{(2)}_{1}\otimes \mathds{1}^{(4)}\cr P^{(1)}_{0}\otimes \mathds{1}^{(3)} & (6,6) & (2,7) & (4\frac{2}{3}, 6\frac{1}{3}) \cr P^{(1)}_{0}\otimes \sigma_{x}^{(3)} & (7,2) & (0,0) & (4\frac{2}{3}, 1\frac{1}{3})\cr P^{(1)}_{1}\otimes \mathds{1}^{(3)} & (6\frac{1}{3}, 4\frac{2}{3}) & (1\frac{1}{3},4\frac{2}{3})  & (5,5)}
\end{align}
\textup{and}
\begin{align}\label{gryschemefake1}
\Gamma'_{Q2}\colon \; \bordermatrix{& P^{(2)}_{0}\otimes \mathds{1}^{(4)} & P^{(2)}_{0}\otimes \sigma_{x}^{(4)} & P^{(2)}_{1}\otimes \mathds{1}^{(4)}\cr P^{(1)}_{0}\otimes \mathds{1}^{(3)} & (2,7) & (6,6) & (3\frac{1}{3}, 6\frac{2}{3}) \cr P^{(1)}_{0}\otimes \sigma_{x}^{(3)} & (0,0) & (7,2) & (2\frac{1}{3}, \frac{2}{3})\cr P^{(1)}_{1}\otimes \mathds{1}^{(3)} & (1\frac{1}{3}, 4\frac{2}{3}) & (6\frac{1}{3},4\frac{2}{3})  & (2\frac{2}{3},4\frac{1}{3})}.
\end{align}
\textup{With Lemma \ref{lemma1} we can show that games (\ref{gryschemefake}) and (\ref{gryschemefake1}) are not isomorphic. Comparing the sets of pure Nash equilibria in both games we find the equilibrium profiles 
\begin{equation}
\{(P^{(1)}_{0} \otimes \mathds{1}^{(3)}, P^{(2)}_{0} \otimes \sigma_{x}^{(4)}), (P^{(1)}_{0} \otimes \sigma_{x}^{(3)}, P^{(2)}_{0} \otimes \mathds{1}^{(4)} ), (P^{(1)}_{1} \otimes \mathds{1}^{(3)}, P^{(2)}_{1} \otimes \mathds{1}^{(4)} )\}
\end{equation}
in the first game and 
\begin{equation}
\{(P^{(1)}_{0} \otimes \mathds{1}^{(3)}, P^{(2)}_{0} \otimes \mathds{1}^{(4)}), (P^{(1)}_{0} \otimes \sigma_{x}^{(3)}, P^{(2)}_{0} \otimes \sigma_{x}^{(4)} )\} 
\end{equation}
in the second one.}
\end{example}
\section{Application of game isomorphism to generalized Marinatto-Weber quantum game scheme}
Additional criteria for a quantum game scheme may have a significant impact on the way how we generalize these schemes. It can be easily seen in the case of the MW scheme \cite{marinatto} (or the refined scheme (\ref{schememoj})), where the sets of unitary strategies are finite. The MW scheme provides us with a quantum model, where the strategy sets consist of the identity operator $\mathds{1}$ and the Pauli operator $\sigma_{x}$. Under this description, what subsets of unitary operators would be suitable for general $n \times m$ games? The case of a 3-element strategy set can be identified with unitary operators $\mathds{1}_{3}$, $C$ and $D$ acting on $\alpha|0\rangle + \beta|1\rangle + \gamma|2\rangle \in \mathbb{C}^3$, where
\begin{equation}\label{operatoryiqbala}
\begin{array}{lll}
\mathds{1}_{3}|0\rangle = |0\rangle, & C|0\rangle = |2\rangle, & D|0\rangle = |1\rangle, \\
\mathds{1}_{3}|1\rangle = |1\rangle, & C|1\rangle = |1\rangle, & D|1\rangle = |0\rangle, \\ 
\mathds{1}_{3}|2\rangle = |2\rangle, & C|2\rangle = |0\rangle, & D|2\rangle = |2\rangle.
\end{array}
\end{equation}
This construction can be found in \cite{iqbal2} and \cite{iqbal3x3}. Another way to generalize the MW scheme was presented in \cite{fracorpolonica}. Having given a strategic-form game, we identify the players' $n$ strategies with $n$ unitary operators $V_{k}$ for $k=0,1,\dots, n-1$. They act on states of the computational basis $\{|0\rangle, |1\rangle, \dots, |n-1\rangle\}$ as follows:
\begin{equation}\label{operatorypolonica}
\begin{array}{l}
V_{0}|i\rangle = |i\rangle,\\
V_{1}|i\rangle = |i+1 ~\mathrm{mod}~ n\rangle,\\
\quad\vdots\\
V_{n-1}|i\rangle = |i+n-1 ~\mathrm{mod}~ n\rangle.
\end{array}
\end{equation}
Both ways to generalize the MW scheme enable us to obtain the classical game. So at this level, neither (\ref{operatoryiqbala}) nor (\ref{operatorypolonica}) is questionable. If we seek other properties, we see that the MW scheme outputs the classical game (or its isomorphic counterpart) when the initial state is one of the computational basis states. Given (\ref{operatoryiqbala}) and (\ref{operatorypolonica}), only the latter case satisfies this condition. Further analysis would show that the MW scheme is invariant with respect to strongly isomorphic input games. It turns out that neither (\ref{operatoryiqbala}) or (\ref{operatorypolonica}) satisfies the isomorphism property.
\begin{example}\label{przyklad2}
\textup{Let us take a look at the following $2\times 3$ bimatrix games:}
\begin{equation}\label{grydlaroznych}
\Gamma\colon \; \bordermatrix{& l & m & r \cr
t & (4,8) & (0,0) & (8,8) \cr 
b & (0,4) & (4,0) & (8,0)
}, \quad \Gamma'\colon \; \bordermatrix{& l' & m' & r' \cr
t' & (0,0) & (4,8) & (8,8) \cr 
b' & (4,0) & (0,4) & (8,0)
}.
\end{equation}
\textup{Consider the MW-type approaches $\Gamma_{Q}$ and $\Gamma'_{Q}$ to games (\ref{grydlaroznych}) according to the following assignement:}
\begin{equation}
\bordermatrix{& l & m & r \cr
t & P_{00} & P_{01} & P_{02} \cr 
b & P_{10} & P_{11} & P_{12}
}, \quad \mbox{\textup{where}}~P_{j_{1}j_{2}} = |j_{1}j_{2}\rangle \langle j_{1}j_{2}|.
\end{equation}
\textup{Then}
\begin{equation}
\Gamma_{Q} = (|\Psi\rangle, (D_{1}, D_{2}), (M_{1}, M_{2})), \quad \Gamma'_{Q} = (|\Psi\rangle, (D'_{1}, D'_{2}), (M'_{1}, M'_{2})),
\end{equation}
\textup{where}
\begin{align}\begin{array}{ll}
M_{1} = 4P_{00} + 8P_{02} + 4P_{11} + 8P_{12}, &  M'_{1} = 4P_{01} + 8P_{02} + 4P_{10} + 8P_{12},\\ 
M_{2} =  8P_{00} + 8P_{02} + 4P_{10}, &  M'_{2} = 8P_{01} + 8P_{02} + 4P_{11}. \end{array}
\end{align}
\textup{Set the initial state $|\Psi\rangle = (1/2)|00\rangle + (\sqrt{3}/2)|12\rangle \in \mathbb{C}^2\otimes \mathbb{C}^3$ and assume first that $D_{1} = D'_{1} = \{\mathds{1}_{2}, \sigma_{x}\}$ and $D_{2} = D'_{2} = \{\mathds{1}_{3}, C, D\}$ as in (\ref{operatoryiqbala}). Determining $\mathrm{tr}\left((U_{1}\otimes U_{2})|\Psi\rangle \langle \Psi|(U^{\dagger}_{1}\otimes U^{\dagger}_{2})M_{i}\right)$ for every $U_{1}\otimes U_{2}\in \{\mathds{1}, \sigma_{x}\} \otimes \{\mathds{1}_{3}, C, D\}$ and $i=1,2$, and doing similar calculations in the case of $M'_{i}$ we obtain}
\begin{equation}\label{pair1}
\Gamma_{Q}\colon \;\bordermatrix{& \mathds{1}_{3} & C & D \cr
\mathds{1}_{2} & (7,2) & (2,5) & (6,0) \cr 
\sigma_{x} & (6,7) & (5,6) & (7,6)
}, \quad \Gamma'_{Q}\colon \;\bordermatrix{& \mathds{1}_{3} & C & D \cr
\mathds{1}_{2} & (6,0) & (5,2) & (7,2) \cr 
\sigma_{x} & (7,6) & (2,0) & (6,7)
}.
\end{equation}
\textup{On the other hand, replacing (\ref{operatoryiqbala}) by (\ref{operatorypolonica}) gives $D_{2} = D'_{2} = \{\mathds{1}_{3}, V_{1}, V_{2}\}$, where}
\begin{equation}\label{operatorypolonica2}
\mathds{1}_{3} =  \left(\begin{array}{ccc} 1 & 0 & 0\\ 0 & 1 & 0 \\ 0 & 0 & 1 \end{array}\right), \quad V_{1} = \left(\begin{array}{ccc} 0 & 0 & 1\\ 1 & 0 & 0 \\ 0 & 1 & 0 \end{array}\right), \quad V_{2} = \left(\begin{array}{ccc} 0 & 1 & 0\\ 0 & 0 & 1 \\ 1 & 0 & 0 \end{array}\right).
\end{equation}
\textup{Then, we have}
\begin{equation}\label{pair2}
\Gamma_{Q}\colon \;\bordermatrix{& \mathds{1}_{3} & V_{1} & V_{2} \cr
\mathds{1}_{2} & (7,2) & (0,3) & (5,2) \cr 
\sigma_{x} & (6,7) & (4,6) & (2,0)
}, \quad \Gamma'_{Q}\colon \;\bordermatrix{& \mathds{1}_{3} & V_{1} & V_{2} \cr
\mathds{1}_{2} & (6,0) & (4,2) & (2,5) \cr 
\sigma_{x} & (7,6) & (0,1) & (5,6)
}.
\end{equation} 
\textup{There is no pure Nash equilibrium in the first game of (\ref{pair1}) and (\ref{pair2}), whereas there are two Nash equilibria in the second games. As a result, each pair of the games do not determine a strong isomorphism.}
\end{example}
Example~\ref{przyklad2} shows that players' strategy sets defined by (\ref{operatoryiqbala}) and (\ref{operatorypolonica2}) need to be revised in order to have a generalized MW scheme invariant with respect to the isomorphism. We shall stick for the moment to considering games (\ref{grydlaroznych}). Let $\{A_{012}, A_{102}, A_{021}, A_{120}, A_{201}, A_{210}\}$ be player 2's strategy set defined to be
\begin{equation}\label{operatorpermutation}
\begin{array}{llllll}
A_{012}|0\rangle = |0\rangle & A_{102}|0\rangle = |1\rangle & A_{021}|0\rangle = |0\rangle & A_{120}|0\rangle = |1\rangle & A_{201}|0\rangle = |2\rangle & A_{210}|0\rangle = |2\rangle,\\
A_{012}|1\rangle = |1\rangle & A_{102}|1\rangle = |0\rangle & A_{021}|1\rangle = |2\rangle & A_{120}|1\rangle = |2\rangle  & A_{201}|1\rangle = |0\rangle & A_{210}|1\rangle = |1\rangle,\\
A_{012}|2\rangle = |2\rangle & A_{102}|2\rangle = |2\rangle & A_{021}|2\rangle = |1\rangle & A_{120}|2\rangle = |0\rangle & A_{201}|2\rangle = |1\rangle & A_{210}|2\rangle = |0\rangle.
\end{array}
\end{equation}
Each $A_{j_{1}j_{2}j_{3}}$ is a permutation matrix that corresponds to a specific permutation $\pi = (0\to j_{1}, 1 \to j_{2}, 2 \to j_{3}) $ of the set $\{0,1,2\}$. Note also that operators (\ref{operatoryiqbala}) and (\ref{operatorypolonica2}) are included in~(\ref{operatorpermutation}). Hence, the MW scheme with (\ref{operatorpermutation}) implies, in particular, the classical game. We now check if it outputs the isomorphic games. Since there are now six operators available for player 2, the resulting game may be written as a $2\times 6$ bimatrix game with entries 
\begin{equation}
\mathrm{tr}\left((U_{1}\otimes U_{2})|\Psi\rangle \langle \Psi|(U^{T}_{1} \otimes U^{T}_{2})M_{i}\right)
\end{equation}
for $U_{1} \in \{\mathds{1}_{2}, \sigma_{x}\}$ and $U_{2} \in \{A_{\pi}\colon \pi - \mbox{permutations of}~\{0,1,2\}\}$. As a result, we obtain
\begin{align}
\Gamma_{Q}\colon \;\bordermatrix{& A_{012} & A_{102} & A_{021} & A_{120} & A_{201} & A_{210} \cr
\mathds{1}_{2} & (7,2) & (6,0) & (4,2) & (0,3) & (5,2) & (2,5) \cr 
\sigma_{x} & (6,7) & (7,6) & (0,1) & (4,6) & (2,0) & (5,6)
}\end{align}
\mbox{and}\\
\begin{align}
\Gamma'_{Q}\colon \;\bordermatrix{& A_{012} & A_{102} & A_{021} & A_{120} & A_{201} & A_{210} \cr
\mathds{1}_{2} & (6,0) & (7,2) & (0,3) & (4,2) & (2,5) & (5,2) \cr 
\sigma_{x} & (7,6) & (6,7) & (4,6) & (0,1) & (5,6) & (2,0)
}.
\end{align}
The games determine the isomorphism $\tilde{f} = (\mathrm{id}_{N}, \tilde{\varphi}_{1}, \tilde{\varphi}_{2})$, where
\begin{align}\begin{split}
&\tilde{\varphi}_{1}= (\mathds{1}_{2} \to \mathds{1}_{2}, \sigma_{x} \to \sigma_{x}),\\
&\tilde{\varphi}_{2} = (A_{012} \to A_{102}, A_{102} \to A_{012}, A_{021} \to A_{120}, A_{120} \to A_{021}, A_{201} \to A_{210}, A_{210} \to A_{201}). \end{split}
\end{align}

Using permutation matrices leads us to formulate another generalized MW scheme. For simplicity, we confine attention to $(n+1)\times (m+1)$ bimatrix games.

Let $S_{n}$ be the set of all permutations $\pi$ of $\{0,1,\dots, n\}$. With each $\pi$ there is associated a permutation matrix $A_{\pi}$,
\begin{equation}\label{operatorpermutation2}
A_{\pi}|i\rangle = |\pi(i)\rangle ~ \mbox{for}~i=0,1,\dots,n. 
\end{equation}
We let $B_{\sigma}$ denote the permutation matrix associated with a permutation $\sigma \in S_{m}$.
Given $(n+1)\times (m+1)$ bimatrix game $\Gamma$ we define 
\begin{equation}\label{propermw}
\Gamma_{Q} = (|\Psi\rangle, (D_{1}, D_{2}), (M_{1}, M_{2})), 
\end{equation}
where 
\begin{align}\label{dopropermw}\begin{split}
&|\Psi\rangle = \sum^{n}_{j_{1}=0}\sum^{m}_{j_{2} = 0}\alpha_{j_{1}j_{2}}|j_{1}j_{2}\rangle \in \mathbb{C}^{n+1}\otimes \mathbb{C}^{m+1}, \;D_{1} = \{A_{\pi}\colon \pi \in S_{n}\}, \;D_{2} = \{B_{\sigma}\colon \sigma \in S_{m}\}, \\
&(M_{1}, M_{2}) = \sum^{n}_{j_{1} = 0}\sum^{m}_{j_{2} = 0}(a_{j_{1}j_{2}},  b_{j_{1}j_{2}})P_{j_{1}j_{2}}. \end{split}
\end{align}
Before stating the main result of this section we start with the observation that the MW scheme remains invariant to numbering of the players. Consider two isomorphic bimatrix games:
\begin{align}\label{gryorderplayers}
\Gamma\colon \; \bordermatrix{& t_{0} & t_{1} & \cdots & t_{m} \cr
s_{0} & (a_{00}, b_{00}) & (a_{01}, b_{01}) & \cdots & (a_{0m}, b_{0m})\cr 
s_{1} & (a_{10}, b_{10}) & (a_{11}, b_{11}) & \cdots & (a_{1m}, b_{1m}) \cr
\;\vdots & \vdots & \vdots & \ddots & \vdots \cr
s_{n} & (a_{n0}, b_{n0}) & (a_{n1}, b_{n1}) & \cdots & (a_{nm}, b_{nm})} 
\end{align}
\mbox{and}
\begin{align}
\Gamma' \colon \; \bordermatrix{& s'_{0} & s'_{1} & \cdots & s'_{n} \cr
t'_{0} & (b_{00}, a_{00}) & (b_{10}, a_{10}) & \cdots & (b_{n0}, a_{n0})\cr
t'_{1} & (b_{01}, a_{01}) & (b_{11}, a_{11}) & \cdots & (b_{n1}, a_{n1}) \cr
\;\vdots & \vdots & \vdots & \ddots & \vdots \cr
t'_{m} & (b_{0m}, a_{0m}) & (b_{1m}, a_{1m}) & \cdots & (b_{nm}, a_{nm})}.
\end{align}
Clearly, the isomorphism is defined by a game mapping $f = \{\pi, \varphi_{1}, \varphi_{2}\}$, where 
\begin{equation}
\pi = (1\to 2, 2 \to 1), \quad \varphi_{1}(s_{j_{1}}) = s'_{j_{1}}, ~ \varphi_{2}(t_{j_{2}}) = t'_{j_{2}}
\end{equation}
for $j_{1} = 0,1,\dots, n$, $j_{2} = 0,1,\dots, m$. The general MW scheme for (\ref{gryorderplayers}) is simply given by (\ref{propermw}).
For $\Gamma'$, we can write 
\begin{equation}\label{propermw2}
\Gamma'_{Q} = (|\Psi'\rangle, (D'_{1}, D'_{2}), (M'_{1}, M'_{2})),
\end{equation}
where
\begin{align}\begin{split}
&|\Psi'\rangle = \sum^{n}_{j_{1}=0}\sum^{m}_{j_{2}=0}\alpha_{j_{1}j_{2}}|j_{2}j_{1}\rangle \in \mathbb{C}^{m+1}\otimes \mathbb{C}^{n+1},\\
&D'_{1} = \{B_{\sigma}, \sigma \in S_{m}\}, D'_{2} = \{A_{\pi}, \pi \in S_{n}\},\\
&M'_{1} = \sum^{n}_{j_{1}=0}\sum^{m}_{j_{2} = 0}b_{j_{1}j_{2}}|j_{2}j_{1}\rangle \langle j_{2}j_{1}|, M'_{2} = \sum^{n}_{j_{1}=0}\sum^{m}_{j_{2} = 0}a_{j_{1}j_{2}}|j_{2}j_{1}\rangle \langle j_{2}j_{1}|. \end{split}
\end{align}
Games determined by (\ref{propermw}) and (\ref{propermw2}) are then isomorphic. To prove this, let $\tilde{f} = (\pi, \tilde{\varphi}_{1}, \tilde{\varphi}_{2})$ be a game mapping such that 
\begin{equation}
\pi = (1\to 2, 2 \to 1), \quad \tilde{\varphi}_{1}\colon D_{1} \to D'_{2},~ \tilde{\varphi}_{1}(A_{\pi}) = A_{\pi}, \quad \tilde{\varphi}_{2}\colon D_{2} \to D'_{1},~ \tilde{\varphi}_{2}(B_{\sigma}) = B_{\sigma}. 
\end{equation}
On account of Definition~\ref{izo} we have
\begin{equation}
\tilde{f}(A_{\pi} \otimes B_{\sigma}) = \varphi_{2}(B_{\sigma})\otimes \varphi_{1}(A_{\pi}) = B_{\sigma} \otimes A_{\pi}.
\end{equation}
As a result,
\begin{align}
u'_{\pi(1)}(\tilde{f}(A_{\pi} \otimes B_{\sigma})) &= u'_{2}(\tilde{f}(A_{\pi} \otimes B_{\sigma})) \nonumber\\
&=\mathrm{tr}\left(\tilde{f}(A_{\pi} \otimes B_{\sigma})|\Psi'\rangle \langle \Psi'|\tilde{f}(A_{\pi} \otimes B_{\sigma})^TM'_{2}\right) \nonumber \\
&=\mathrm{tr}\left((\varphi_{2}(B_{\sigma})\otimes \varphi_{1}(A_{\pi}))|\Psi'\rangle \langle \Psi'|(\varphi_{2}(B_{\sigma})^T\otimes \varphi_{1}(A_{\pi})^T)M'_{2}\right)\nonumber\\
&=\mathrm{tr}\left((B_{\sigma}\otimes A_{\pi})|\Psi'\rangle \langle \Psi'|(B^T_{\sigma}\otimes A^T_{\pi})\sum^{n}_{j_{1}=0}\sum^{m}_{j_{2} = 0}a_{j_{1}j_{2}}|j_{2}j_{1}\rangle \langle j_{2}j_{1}|\right)\nonumber\\
&=\mathrm{tr}\left((A_{\pi}\otimes B_{\sigma})|\Psi\rangle \langle \Psi|(A^T_{\pi}\otimes B^T_{\sigma})\sum^{n}_{j_{1}=0}\sum^{m}_{j_{2} = 0}a_{j_{1}j_{2}}|j_{1}j_{2}\rangle \langle j_{1}j_{2}|\right)\nonumber \\
&=u_{1}(A_{\pi}\otimes B_{\sigma}).
\end{align}
By a similar argument, we can show that $u'_{\pi(2)}(\tilde{f}(A_{\pi} \otimes B_{\sigma})) = u_{2}(A_{\pi}\otimes B_{\sigma})$.
We can now formulate the following proposition:
\begin{proposition}\label{prop}
Assume that $\Gamma$ and $\Gamma'$ are strongly~isomorphic bimatrix games and $\Gamma_{Q}$ and $\Gamma'_{Q}$ are the corresponding quantum games defined by~(\ref{propermw}). Then $\Gamma_{Q}$ and $\Gamma'_{Q}$ are strongly isomorphic.
\end{proposition}
\begin{proof}
Let $\Gamma$ and $\Gamma'$ be bimatrix games of dimension $n\times m$ and let $f\colon \Gamma \to \Gamma'$,  $f = (\eta, \varphi_{1}, \varphi_{2})$ be the strong isomorphism. Since the MW scheme is invariant to numbering of the players, there is no loss of generality in assuming $\eta = \mathrm{id}_{N}$. Now, it follows from Definition \ref{izo} that games $\Gamma$ and $\Gamma'$ differ in the order of players' strategies. Let us identify players' strategies (i.e., rows and columns) in $\Gamma$  with sequences $(0,1, \dots, n)$ and $(0,1, \dots, m)$, respectively. Then, we denote by $\pi^{*}$ and $\sigma^{*}$ the permutations of the sets $\{0,1, \dots, n\}$ and $\{0,1,\dots, m\}$ associated with the order of strategies in game $\Gamma'$. A trivial verification shows that the payoff operator $M'_{i}$ in $\Gamma'_{Q}$ may be written as
\begin{equation}
M'_{i} = \left(A_{\pi^{*}}\otimes B_{\sigma^{*}}\right)M_{i}\left(A_{\pi^{*}}\otimes B_{\sigma^{*}}\right)^T,
\end{equation}
where $A_{\pi^{*}}$ and $B_{\sigma^{*}}$ are the permutation matrices corresponding to $\pi^{*}$ and $\sigma^{*}$. Define a game mapping $\tilde{f} = (\mathrm{id}_{N}, \tilde{\varphi}_{1}, \tilde{\varphi}_{2})$, where
\begin{align}\begin{split}
&\tilde{\varphi}_{1}\colon \{A_{\pi}\colon \pi \in S_{n}\} \to \{A_{\pi}\colon \pi \in S_{n}\}, \; \varphi_{1}(A_{\pi}) = A_{\pi^{*}}A_{\pi}; \cr
&\tilde{\varphi}_{2}\colon \{B_{\sigma}\colon \sigma \in S_{m}\} \to \{B_{\sigma}\colon \sigma \in S_{m}\}, \; \varphi_{2}(B_{\sigma}) = B_{\sigma^{*}}B_{\sigma}. \end{split}
\end{align}
Hence, $\tilde{f}$ maps $A_{\pi} \otimes B_{\sigma}$ to $A_{\pi^*}A_{\pi} \otimes B_{\sigma^*}B_{\sigma}$. Thus, we obtain
\begin{align}\label{mac6}
u'_{i}(\tilde{f}(A_{\pi} \otimes B_{\sigma})) &=\mathrm{tr}\left(\tilde{f}(A_{\pi}\otimes B_{\sigma})|\Psi\rangle \langle \Psi| \left(\tilde{f}(A_{\pi}\otimes B_{\sigma})\right)^TM'_{i}\right)\nonumber \\  &= \mathrm{tr}\left((A_{\pi^*}A_{\pi} \otimes B_{\sigma^*}B_{\sigma})|\Psi\rangle \langle \Psi|(A_{\pi^*}A_{\pi}\otimes B_{\sigma^*}B_{\sigma})^T\left(A_{\pi^{*}}\otimes B_{\sigma^{*}}\right)M_{i}\left(A_{\pi^{*}}\otimes B_{\sigma^{*}}\right)^T\right) \nonumber \\  &=\mathrm{tr}\left(\left(A_{\pi^{*}}\otimes B_{\sigma^{*}}\right)^T(A_{\pi^*}A_{\pi} \otimes B_{\sigma^*}B_{\sigma})|\Psi\rangle \langle \Psi|(A_{\pi^*}A_{\pi}\otimes B_{\sigma^*}B_{\sigma})^T\left(A_{\pi^{*}}\otimes B_{\sigma^{*}}\right)M_{i}\right)\nonumber \\
&=\mathrm{tr}\left(\left(A^T_{\pi^{*}}A_{\pi^{*}}A_{\pi}\otimes B^T_{\sigma^{*}}B_{\sigma^{*}}B_{\sigma}\right)|\Psi\rangle \langle \Psi|\left(A^T_{\pi}A^T_{\pi^{*}}A_{\pi^{*}}\otimes B^T_{\sigma}B^T_{\sigma^{*}}B_{\sigma^*}\right)M_{i}\right)\nonumber \\  &=\mathrm{tr}\left(\left(A_{\pi}\otimes B_{\sigma}\right)|\Psi\rangle \langle \Psi|\left(A_{\pi}\otimes B_{\sigma}\right)^{T}M_{i}\right) = u_{i}(A_{\pi}\otimes B_{\sigma}). 
\end{align}
This finishes the proof.
\end{proof}
Note that operators (\ref{operatorpermutation2}) come down to $\mathds{1}$ and $\sigma_{x}$ for $n=1$. Therefore, the original MW scheme preserves the isomorphism. The same conclusion can be drawn for the refined MW scheme (\ref{schememoj}).
\begin{corollary}
If $\Gamma$ and $\Gamma'$ are strongly isomorphic bimatrix games and $\Gamma_{Q}$ and $\Gamma'_{Q}$ are the corresponding games defined by (\ref{schememoj}). Then $\Gamma_{Q}$ and $\Gamma'_{Q}$ are strongly isomorphic.
\end{corollary}
\begin{proof}
Let $\Gamma$ and $\Gamma'$ be strongly isomorphic $2\times 2$ bimatrix games. By Proposition~\ref{prop} there exists a strong isomorphism $\tilde{f} = (\mathrm{id}_{\{1,2\}}, \tilde{\varphi}_{1}, \tilde{\varphi}_{2})$ between the games $\Gamma_{Q}$ and $\Gamma'_{Q}$ played according to (\ref{propermw})-(\ref{dopropermw}). Given the quantum approach (\ref{schememoj}) to $\Gamma$ and $\Gamma'$ we define $\tilde{g} = (\mathrm{id}_{\{1,2\}}, \tilde{\xi}_{1}, \tilde{\xi}_{2})$, where 
\begin{align}\begin{split}
&\tilde{\xi}_{1}\colon S_{1} \to S_{1}, \; \tilde{\xi}_{1}\left(P^{(1)}_{i}\otimes U^{(3)}_{j}\right) = P^{(1)}_{i}\otimes \tilde{\varphi}_{1}\left(U^{(3)}_{j}\right),\\
&\tilde{\xi}_{2}\colon S_{2} \to S_{2}, \; \tilde{\xi}_{2}\left(P^{(2)}_{k}\otimes U^{(4)}_{l}\right) = P^{(2)}_{k}\otimes \tilde{\varphi}_{2}\left(U^{(4)}_{l}\right). \end{split}
\end{align}
Now, we have
\begin{align}\label{eqwniosek}
&u'_{i}\left(\tilde{g}\left(P^{(1)}_{i}\otimes P^{(2)}_{k}\otimes U^{(3)}_{j}\otimes U^{(4)}_{l}\right)\right) \nonumber\\
&\quad = \mathrm{tr}\left(\tilde{g}\left(P^{(1)}_{i}\otimes P^{(2)}_{k}\otimes U^{(3)}_{j}\otimes U^{(4)}_{l}\right)H\tilde{g}\left(P^{(1)}_{i}\otimes P^{(2)}_{k}\otimes U^{(3)}_{j}\otimes U^{(4)}_{l}\right)^TM'_{1}\right)\nonumber \\
&\quad =  \mathrm{tr}\left(P^{(1)}_{i} \otimes P^{(2)}_{k}\otimes \tilde{\varphi}_{1}\left(U^{(3)}_{j}\right) \otimes \tilde{\varphi}_{2}\left(U^{(4)}_{l}\right) H P^{(1)}_{i} \otimes P^{(2)}_{k}\otimes \tilde{\varphi}_{1}\left(U^{(3)}_{j}\right)^T \otimes \tilde{\varphi}_{2}\left(U^{(4)}_{l}\right)^TM'_{i}\right)\nonumber\\
&\quad =  \mathrm{tr}\left(P^{(1)}_{i} \otimes P^{(2)}_{k}\otimes \tilde{f}\left(U^{(3)}_{j} \otimes U^{(4)}_{l}\right) H P^{(1)}_{i} \otimes P^{(2)}_{k}\otimes \tilde{f}\left(U^{(3)}_{j} \otimes U^{(4)}_{l}\right)^T M'_{i}\right).
\end{align}
For fixed $P_{i}\otimes P_{k}$, we can write the right side of~(\ref{eqwniosek}) in the form
\begin{equation}
\mathrm{tr}\left(|ik\rangle \langle ik| \otimes \tilde{f}\left(U^{(3)}_{j} \otimes U^{(4)}_{l}\right) \rho\tilde{f}\left(U^{(3)}_{j} \otimes U^{(4)}_{l}\right)^TM'_{i}\right), \; \rho = \begin{cases}|\Psi\rangle \langle \Psi|, & (i,j) = (1,1);\\ |00\rangle \langle 00|, & (i,j) \ne (1,1).\end{cases}
\end{equation}
By reasoning similar to~(\ref{mac6}) we conclude that
\begin{align}
&u'_{i}\left(\tilde{g}\left(P^{(1)}_{i}\otimes P^{(2)}_{k}\otimes U^{(3)}_{j}\otimes U^{(4)}_{l}\right)\right)\nonumber\\
&\quad = \mathrm{tr}\left(|ik\rangle \langle ik| \otimes \tilde{f}\left(U^{(3)}_{j} \otimes U^{(4)}_{l}\right) \rho\tilde{f}\left(U^{(3)}_{j} \otimes U^{(4)}_{l}\right)^TM'_{i}\right)\nonumber\\
&\quad = \mathrm{tr}\left(|ik\rangle \langle ik| \otimes U^{(3)}_{j} \otimes U^{(4)}_{l} \rho \left(U^{(3)}_{j} \otimes U^{(4)}_{l}\right)^TM'_{i}\right)\nonumber\\
&\quad = u_{i}\left(P^{(1)}_{i}\otimes P^{(2)}_{k}\otimes U^{(3)}_{j}\otimes U^{(4)}_{l}\right).
\end{align}
We have thus proved that $\Gamma_{Q}$ and $\Gamma'_{Q}$ are isomorphic.
\end{proof}
It is worth noting that the converse may not be true. Given isomorphic games $\Gamma_{Q}$ and $\Gamma'_{Q}$, the input games $\Gamma$ and $\Gamma'$ may not determine the strong isomorphism. Indeed, bimatrix games
\begin{equation}\label{ostatniegames}
\bordermatrix{& l & r \cr t & (3,1) & (0,0) \cr b & (0,0) & (1,3)} \quad \mbox{and} \quad \bordermatrix{& l' & r \cr t' & (4,0) & (0,0) \cr b' & (0,0) & (0,4)}
\end{equation}
are not strongly isomorphic. However the MW approach (with the initial state $(|00\rangle + |11\rangle)/\sqrt{2}$) to each one of (\ref{ostatniegames}) implies the same output game given by
\begin{equation}
\bordermatrix{& \mathds{1} & \sigma_{x} \cr \mathds{1} & (3,1) & (0,0) \cr \sigma_{x} & (0,0) & (1,3)}
\end{equation}
\section{Conclusions}
The theory of quantum games does not provide us with clear definitions of how a quantum game should look like. In fact, only one condition is taken into consideration. A quantum game scheme is merely required to generalize the classical game. As a result, this allows us to define a quantum game scheme in many different ways. However, a wide variety of techniques to describe a game in the quantum domain can imply different quantum game results. Therefore, it would be convenient to specify that some quantum schemes work under some further restrictions. We have been working under the assumption that a quantum scheme is invariant with respect to isomorphic transformations of an input game. We have shown that this requirement may be essential tool in defining a quantum scheme. The protocol that replicates classical correlated equilibria is an example that does not satisfy our criterion. The refined definition for a quantum game scheme may also be useful to generalize protocols. Our work has shown that dependence of local unitary operators in the MW scheme on the number of strategies in a classical game is not linear. In fact, the generalized approach to $n\times m$ bimatrix game can be identified with a game of dimension $n! \times m!$.

\section*{Acknowledges}
This work was supported by the Ministry of Science and Higher Education in Poland under the Project Iuventus Plus IP2014 010973 in the years 2015--2017.


\end{document}